\documentclass[12pt]{article}

\usepackage[utf8]{inputenc}
\usepackage{amsthm}

\newtheorem{theorem}{Theorem}

\theoremstyle{remark}
\newtheorem*{remark}{Remark}

\usepackage{authblk} 

\usepackage[margin=1in]{geometry} 
\usepackage{braket}
\usepackage{amsmath,amsthm,amssymb}
\usepackage{graphicx}% Include figure files
\usepackage[makeroom]{cancel}
\usepackage{empheq}
\usepackage{tensor}
\usepackage{subfig}
\usepackage{tikz-cd}
\usepackage[
    colorlinks=true,
    pdfborder={0 0 0},
    linkcolor=blue!70!black,
    citecolor =blue!70!black,
    urlcolor = blue!70!black
]{hyperref}
\usepackage{cleveref}

\begin{document}

\title{Blázquez-Salcedo--Knoll--Radu Wormholes Are Not Solutions to the Einstein--Dirac--Maxwell Equations}
\author{Daine L. Danielson, Gautam Satishchandran, Robert M. Wald, Robert J. Weinbaum \\Enrico Fermi Institute and Department of Physics\\ The University of Chicago\\ 5640 South Ellis Avenue, Chicago, Illinois 60637, USA}
%\affil{Enrico Fermi Institute and Department of Physics\\ The University of Chicago\\ 5640 South Ellis Avenue, Chicago, IL 60637, USA.}
\date{}                     %% if you don't need date to appear
\setcounter{Maxaffil}{0}
\renewcommand\Affilfont{\itshape\small}

\maketitle

\begin{abstract}

Recently, Blázquez-Salcedo, Knoll, and Radu (BSKR) have given a class of static, spherically symmetric, traversable wormhole spacetimes with Dirac and Maxwell fields. The BSKR wormholes are obtained by joining a classical solution to the Einstein--Dirac--Maxwell (EDM) equations on the ``up'' side of the wormhole ($r \geq 0$) to a corresponding solution on the ``down'' side of the wormhole ($r \leq 0$). However, it can be seen that the BSKR metric fails to be $C^3$ on the wormhole throat at $r=0$. We prove that if the matching were done in such a way that the resulting spacetime metric, Dirac field, and Maxwell field composed a solution to the EDM equations in a neighborhood of $r=0$, then all of the fields would be smooth at $r=0$ in a suitable gauge. Thus, the BSKR wormholes cannot be solutions to the EDM equations. The failure of the BSKR wormholes to solve the EDM equations arises both from the failure of the Maxwell field to satisfy the required matching conditions (which implies the presence of an additional shell of charged matter at $r=0$) and, more significantly, from the failure of the Dirac field to satisfy required matching conditions (which implies the presence of a spurious source term for the Dirac field at $r=0$).
    
\end{abstract}

\section{Introduction}

In a recent paper \cite{Blazquez-Salcedo:2020czn}, Blázquez-Salcedo, Knoll, and Radu (BSKR) have provided examples of traversable wormhole spacetimes, which are claimed to be classical solutions to the Einstein--Dirac--Maxwell (EDM) equations. The existence of traversable wormhole solutions without unphysical matter would be of great interest, as it would open new possibilities for the topology of spacetime and for causal connections between different regions of spacetime. The existence of a traversable wormhole would require a violation of the averaged null energy condition \cite{Friedman:1993ty,Galloway_1995}, and such a violation would be of considerable interest in its own right. 

The BSKR wormhole spacetimes have a static, spherically symmetric metric of the form
\begin{equation}
    ds^2 = -F^2_0 (r) dt^2 + F_1^2 (r) dr^2 + F^2_2(r) d \Omega^2
    \label{whmetric}
\end{equation}
where $F_2(r)$ is taken to be of the form
\begin{equation}
    F_2(r) = \sqrt{(r^2 + r_0^2)}.
    \label{F2}
\end{equation}
The wormhole spacetimes are constructed by finding smooth solutions of the EDM equations separately in the regions $r \geq 0$ and $r \leq 0$ and then joining these solutions along the timelike hypersurface $\Sigma$ at $r=0$. In order that the resulting metric and fields solve the EDM equations without the presence of spurious sources at $\Sigma$, it is necessary that suitable matching conditions be satisfied. For the metric, the matching conditions \cite{Israel:1966rt} are that the induced metric of $\Sigma$ and the extrinsic curvature of $\Sigma$ agree. For the electromagnetic field, the matching conditions are that we can choose a gauge in which the $4$-vector potential $A_\mu$ and its normal derivative match. For the Dirac field, the matching conditions are that if we work in a tetrad that is continuous at $\Sigma$ and in an electromagnetic gauge where $A_\mu$ is continuous at $\Sigma$, the Dirac field components must match. Failure to satisfy the matching of the induced metric would yield a spacetime for which the normal derivative of the metric would have a delta-function contribution at $\Sigma$, and the resulting Einstein tensor could not even be interpreted distributionally \cite{PhysRevD.36.1017}. Failure to satisfy any of the other matching conditions would correspond to the presence of spurious distributional source terms at $\Sigma$ in the EDM equations.

Satisfaction of the matching conditions directly requires that, in a suitable gauge, the spacetime metric and vector potential must be $C^1$ at $\Sigma$ and the Dirac field must be $C^0$ at $\Sigma$. We will show in Sec. \ref{smoothness} that the satisfaction of the EDM equations for $r \geq 0$ and $r \leq 0$ then implies that, in a suitable gauge, all derivatives of the metric, Maxwell field, and Dirac field must match on $\Sigma$, so for any solution of the EDM equations without spurious sources on $\Sigma$, all of these fields must be smooth ($C^\infty$). However, for all of the BSKR wormholes, the metric fails to be $C^3$ on $\Sigma$. We therefore conclude that none of the BSKR wormholes can be solutions to the EDM equations; i.e., they all must contain spurious sources on $\Sigma$.

BSKR properly impose the matching conditions for the metric. However, as we shall see in Sec. \ref{dirac}, they did not impose proper matching conditions on the electromagnetic field, resulting in the presence of an additional charged shell at $r=0$. It is conceivable that such a charged shell could be modeled by the presence of other physical charged matter, although it is not obvious that traversable wormhole solutions could be obtained in this way, since any physically acceptable charged matter also would contribute to the matter stress-energy tensor. However, a failure of the matching conditions for the Dirac field would be more serious, since the Dirac field does not have any known physical sources. Therefore, it is important to check if the matching conditions for the Dirac field are satisfied. This is not entirely straightforward to analyze, since BSKR use a tetrad whose radial vector $e^a_3$ points in the positive $r$ direction for $r \geq 0$ and in the negative $r$ direction for $r \leq 0$. Thus, their tetrad is discontinuous at $r=0$, and this discontinuity must be taken into account when considering the matching conditions. In Sec. \ref{dirac}, we obtain the proper matching conditions required for the continuity of the Dirac field. We find that these matching conditions are not satisfied by the BSKR wormholes.

In summary, the BSKR wormholes contain both shells of charged matter and, more significantly, spurious sources for the Dirac field at the wormhole throat at $r=0$. Therefore, they are not solutions to the EDM equations.

\medskip

\noindent
{\it Notation.}---We will use lowercase latin letters from the early part of the alphabet (i.e., $a,b,\dots$) to denote abstract spacetime indices, e.g., the spacetime metric will be denoted as $g_{ab}$. We will use greek letters from the middle part of the alphabet (i.e., $\mu, \nu, \dots$) to denote coordinate components of tensors. We also use mid-alphabet greek letters to enumerate tetrad vectors (e.g., a tetrad will be denoted as $\{e^a_\mu \}$, with $\mu = 0,1,2,3$). We will use lowercase latin indices from the mid-part of the alphabet (i.e., $i,j,\dots$) to denote the non-normal components of tensors in Gaussian normal coordinates based on the timelike hypersurface $\Sigma$ at $r=0$. (For Gaussian normal coordinates based on a spacelike hypersurface, these would correspond to spatial components, but they correspond to the nonradial components in our case.) Finally, we will use uppercase latin indices to denote Weyl spinors and use lowercase greek indices from the early alphabet (i.e., $\alpha, \beta, \dots$) to denote components of Dirac spinors. Thus $\phi^A$ denotes a Weyl spinor, and $\Psi_\alpha$ denotes a Dirac spinor.

\section{Smoothness of Solutions at $r=0$}
\label{smoothness}

In this section, we consider spacetimes that are obtained by gluing solutions along a noncharacteristic (i.e., timelike or spacelike) boundary. Our arguments and results are extremely general, but to keep the discussion simple, we will restrict consideration to the EDM system. 

Suppose we are given a smooth ($C^{\infty}$) solution $(g^+_{ab}, A^+_a, \Psi^+_\alpha)$ to the EDM equations on a manifold $M^+$. Suppose that a boundary $\Sigma^+$ can be attached to $M^+$ such that $M^+ \cup \Sigma^+$ is a manifold with boundary. Suppose that $(g^+_{ab}, A^+_a, \Psi^+_\alpha)$ can be smoothly extended to $\Sigma^+$ and that $\Sigma^+$ is everywhere noncharacteristic with respect to $g^+_{ab}$; i.e., it is either everywhere spacelike or everywhere timelike. In the case of BSKR wormholes, $M^+$ would correspond to the region $r > 0$, and $\Sigma^+$ would correspond to the timelike hypersurface $r = 0$. For definiteness, we will assume in the following that $\Sigma^+$ is timelike. 

Now suppose we also are given another smooth solution $(g^-_{ab}, A^-_a, \Psi^-_\alpha)$ to the EDM equations on a manifold $M^-$ that smoothly extends to a timelike boundary $\Sigma^-$. For the BSKR wormholes, $M^-$ would correspond to the region $r < 0$, and $\Sigma^-$ would correspond to the hypersurface $r = 0$. If we identify $\Sigma^+$ with $\Sigma^-$ and denote the identified surface as $\Sigma$, we will obtain the enlarged spacetime $M = M^+ \cup \Sigma \cup M^-$ with fields $(g_{ab}, A_a, \Psi_\alpha)$. The BSKR wormhole spacetimes are constructed in this manner. We wish to investigate the conditions under which the fields $(g_{ab}, A_a, \Psi_\alpha)$ satisfy the EDM equations. In the case where they do satisfy the EDM equations, we also wish to investigate their smoothness properties. Obviously, satisfaction of the EDM equations and smoothness needs to be investigated only in an arbitrarily small neighborhood of $\Sigma$, since we have assumed that $(g_{ab}, A_a, \Psi_\alpha)$ is a smooth solution of the EDM equations on $M^+$ and $M^-$.

It is useful to make appropriate gauge choices for our original solutions $(g^+_{ab}, A^+_a, \Psi^+_\alpha)$ and $(g^-_{ab}, A^-_a, \Psi^-_\alpha)$ so that any nonsmoothness in the matching will not be a gauge artifact. It is very convenient to use Gaussian normal coordinates $(s_+, x_+^i)$ on $M^+$ in a neighborhood of $\Sigma^+$. Gaussian normal coordinates are defined by choosing coordinates $x_+^i$ on $\Sigma^+$ and extending them off of $\Sigma^+$ by keeping them constant on normal geodesics. We then define $s_+$ to be the proper distance from $\Sigma^+$ along each normal geodesic. It follows that in Gaussian normal coordinates, we have $g^+_{ss} = 1$ and $g^+_{si} = 0$. Similarly, we choose Gaussian normal coordinates on $M^-$ in a neighborhood of $\Sigma^-$, except that in this case we take $s_-$ to be minus the proper distance from $\Sigma^-$ along the normal geodesic. 

For the electromagnetic field, it is very convenient to work in a gauge on $M^+$ where $A^+_s = 0$, i.e., $A^{+}_a (\partial/\partial s)^a = 0$. This gauge can be achieved starting in an arbitrary gauge by choosing any function $\chi_0(x_+^i)$ on $\Sigma^+$ and solving $\partial \chi(s, x_+^i)/\partial s = A^+_s$ with the initial condition $\chi(0, x_+^i) = \chi_0(x_+^i)$. The gauge transformed potential $A'^+_a = A^+_a - \nabla_a \chi$ then satisfies the desired gauge condition. Similarly, we choose $A^-_s = 0$ in a neighborhood of $\Sigma^-$ in $M^-$. 

Finally, for the Dirac field, we must specify a tetrad in order to define its components. We choose a tetrad in $M^+$ by choosing an orthonormal triad, $\{ e^{+i}_{0}, e^{+i}_{1}, e^{+i}_{2} \}$,  on $\Sigma^+$ tangent to $\Sigma^+$ and supplementing it with $e^{+a}_{3} = (\partial/\partial s)^a$. We then propagate this tetrad into $M^+$ by parallel transport along the normal geodesics. We choose a tetrad in $M^-$ in the same manner. Note that since $s$ takes negative values in $M^-$, $e^{-a}_{3} = (\partial/\partial s)^a$ points toward $\Sigma^-$ in $M^-$.

As stated above, the spacetime $M$ is obtained by identifying $\Sigma^+$ and $\Sigma^-$. Since our gauge choices above do not place any restrictions on the choices of coordinates $x_+^i$ and $x_-^i$ on $\Sigma^+$ and $\Sigma^-$, we may assume without loss of generality that the identification is such that $x_+^i = x_-^i$. We then may drop the plus and minus subscripts on $x^i$. We then have the following necessary conditions for $(g_{ab}, A_a, \Psi_\alpha)$ to be a solution of the EDM equations on $M$.

First, in order to satisfy Einstein's equation, it is essential that the induced metric on $\Sigma^+$ and $\Sigma^-$ match, i.e., 
\begin{equation}
   g^+_{ij} \big\vert_{s=0} = g^-_{ij} \big\vert_{s=0} \, .
\label{metmatch}
\end{equation}
If this were not the case, the metric would be discontinuous on $\Sigma$. In that case, $\partial g^+_{ij}/\partial s$ would have a delta-function singularity and, as previously mentioned, nonlinear terms in the Einstein tensor involving this quantity could not even be defined %\cite{Geroch:1986jjl}
\cite{PhysRevD.36.1017}. It also is necessary for a solution to the Einstein portion of the EDM equations that the extrinsic curvatures of $\Sigma^+$ and $\Sigma^-$ match, i.e., that
\begin{equation}
    \frac{\partial g^+_{ij}}{\partial s} \bigg\vert_{s=0} = \frac{\partial g^-_{ij}}{\partial s} \bigg\vert_{s=0} \, .
    \label{ecmatch}
\end{equation}
If this condition were not satisfied, it would give rise to a $\delta$-function contribution to the Einstein tensor, corresponding to the presence of an additional shell of matter \cite{Israel:1966rt}.

Maxwell's equations are first order differential equations involving the gauge invariant field strength tensor $F_{\mu \nu}$. Given the matching of the metric and extrinsic curvature as specified in the previous paragraph, it is necessary that 
\begin{equation}
    F^+_{\mu \nu} \big\vert_{s=0} = F^-_{\mu \nu} \big\vert_{s=0}
\end{equation}
since any failure of this matching to hold would give rise to a $\delta$-function contribution to Maxwell's equations, which would correspond to the presence of an additional shell of electric and/or magnetic charge and/or current. If this matching condition holds, then, in particular, the tangential components $F^+_{ij}$ and $F^-_{ij}$ match. Since our gauge condition $A_s = 0$ allows the freedom to perform any $s$-independent gauge transformation, we may use such gauge freedom to require matching of the vector potentials on $\Sigma$
\begin{equation}
   A^+_{i} \big\vert_{s=0} = A^-_{i} \big\vert_{s=0} \, .
   \label{vpmatch}
\end{equation}
The matching of the components $F^+_{si}$ and $F^-_{si}$ at $s=0$ then requires
\begin{equation}
    \frac{\partial A^+_{i}}{\partial s} \bigg\vert_{s=0} = \frac{\partial A^-_{i}}{\partial s} \bigg\vert_{s=0} \, .
    \label{vpmatch2}
\end{equation}

If Eq. (\ref{metmatch}) holds, we may choose the triad on $\Sigma^+$ to match the triad on $\Sigma^-$. It then follows that for a solution to the EDM equations, the Dirac field components must match on $\Sigma$, i.e.
\begin{equation}
   \Psi^+_{\alpha} \big\vert_{s=0} = \Psi^-_{\alpha} \big\vert_{s=0} \, 
   \label{dirmatch}
\end{equation}
since otherwise there would be a spurious delta-function source term in the Dirac equation.

In summary, we have just shown that a {\em necessary} condition for $(g_{ab}, A_a, \Psi_\alpha)$ to satisfy the EDM equations on $M$ is that in Gaussian normal coordinates associated with $\Sigma$ and in the gauge $A_s = 0$, further gauge choices can be made, if necessary, so that Eqs. (\ref{metmatch})--(\ref{dirmatch}) hold. The main result of this section is the following theorem:

\medskip

\noindent
\begin{theorem} Let $(g_{ab}, A_a, \Psi_\alpha)$ be the fields on $M$ obtained by gluing together solutions of the EDM equations on $M^+$ and $M^-$ in the manner described above. Suppose that in Gaussian normal coordinates on $\Sigma$ and in the gauge $A_s = 0$ the fields satisfy the matching conditions (\ref{metmatch})--(\ref{dirmatch}). Then $(g_{ab}, A_a, \Psi_\alpha)$ is a solution to the EDM equations on $M$. Furthermore, all of these fields are smooth ($C^\infty$).
\end{theorem}

\begin{proof} We will show that, with our gauge choices, the fields $(g_{ab}, A_a, \Psi_\alpha)$ are smooth on $\Sigma$. Once smoothness is established, it follows immediately by continuity that they satisfy the EDM equations on $\Sigma$ since, by construction, these fields satisfy the EDM equations everywhere off of $\Sigma$.

By hypothesis, $(g^+_{ab}, A^+_a, \Psi^+_\alpha)$ is smooth for $s \geq 0$ and $(g^-_{ab}, A^-_a, \Psi^-_\alpha)$ is smooth for $s \leq 0$, so the only way $(g_{ab}, A_a, \Psi_\alpha)$ could fail to be smooth is if these quantities or their derivatives with respect to $s$ fail to match at $s=0$. Eqs.(\ref{metmatch}) and (\ref{ecmatch}) require matching of the metric and its first $s$ derivative at $s=0$, so the metric is at least $C^1$. Similarly, Eqs. (\ref{vpmatch}) and (\ref{vpmatch2}) imply that the vector potential is at least $C^1$ and Eq. (\ref{dirmatch}) implies that the Dirac field is at least $C^0$. As we shall now show, smoothness of these quantities then follows from the basic form of the EDM equations. 

In order to write the Dirac equation, we must make a choice of tetrad. We will use the tetrad $\{ e^a_\mu \}$ introduced above. Since $e_3^a = (\partial/\partial s)^a$ and the $s$ component of each of the other tetrad vectors vanishes, we need only be concerned with $e^i_\mu$ for $\mu = 0,1,2$. The parallel propagation evolution law takes the form
\begin{equation}
    \frac{\partial e^i_\mu}{\partial s} = f^{Ti}_{\mu} \left[g_{kl}, \frac{\partial g_{kl}}{\partial s}, e^j_\nu  \right].
    \label{teteq}
\end{equation}
Here $f^{Ti}_{\mu}$ is a smooth function of the indicated variables together with finitely many of their $x^k$ derivatives.
The Dirac equation then takes the form
\begin{equation}
    \frac{\partial \Psi_\alpha}{\partial s} = f^D_{\alpha} \left[A_k, \Psi_\alpha, e^i_\mu, \frac{\partial e^i_\mu}{\partial s} \right]
    \label{direq}
\end{equation}
where $f^{D}_{\alpha}$ is a smooth function of the indicated variables together with finitely many of their $x^k$ derivatives.
Note that the metric does not appear on the right side, since it can be reconstructed from the tetrad via
\begin{equation}
    g^{ij} = \sum_{\mu, \nu =0}^2 \eta^{\mu \nu} e^i_\mu e^j_\nu.
\end{equation}
Einstein's equation in Gaussian normal coordinates takes the form
\begin{equation}
    \frac{\partial^2 g_{ij}}{\partial s^2} = f^E_{ij} \left[e^i_\mu, \frac{\partial e^i_\mu}{\partial s}, A_k, \frac{\partial A_{i}}{\partial s}, \Psi_\alpha, \frac{\partial \Psi_{\alpha}}{\partial s} \right]
    \label{eineq}
\end{equation}
where $f^E_{ij}$ is a smooth function of the indicated variables together with finitely many of their $x^k$ derivatives. Finally, Maxwell's equations in the gauge $A_s = 0$ take the form
\begin{equation}
    \frac{\partial^2 A_{i}}{\partial s^2} = f^M_{i} \left[e^i_\mu, \frac{\partial e^i_\mu}{\partial s}, A_k, \frac{\partial A_{i}}{\partial s}, \Psi_\alpha \right].
    \label{maxeq}
\end{equation}

We now have all the ingredients necessary to prove smoothness. The quantities $(g^+_{ab}, A^+_a, \Psi^+_\alpha, e^{+i}_\mu)$ satisfy Eqs. (\ref{teteq})--(\ref{maxeq}) for $s > 0$, whereas $(g^-_{ab}, A^-_a, \Psi^-_\alpha, e^{-i}_\mu)$ satisfy Eqs. (\ref{teteq})--(\ref{maxeq}) for $s < 0$. By hypothesis, the matching conditions (\ref{metmatch})--(\ref{dirmatch}) hold on $\Sigma$. By construction, the tetrad vectors also match on $\Sigma$. It then follows from Eq. (\ref{teteq}) that the normal derivatives of the tetrad vectors match on $\Sigma$, so the tetrad is $C^1$. It then further follows from Eq. (\ref{direq}) and the matching conditions (\ref{vpmatch}) and (\ref{dirmatch}) that 
\begin{equation}
    \frac{\partial \Psi^+_{\alpha}}{\partial s} \bigg\vert_{s=0} = \frac{\partial \Psi^-_{\alpha}}{\partial s} \bigg\vert_{s=0} \, .
    \label{dirmatch2}
\end{equation}
Thus, $\Psi_\alpha$ is $C^1$. It then follows immediately from Eqs. (\ref{eineq}) and (\ref{maxeq}) together with our matching conditions that the metric and vector potential are $C^2$ at $s=0$. We now take an $s$ derivative of Eqs. (\ref{teteq}), (\ref{direq}), (\ref{eineq}), and (\ref{maxeq}) and repeat the argument to conclude that the tetrad vectors and Dirac field are $C^2$ and the metric and vector potential are $C^3$. By induction, all fields are $C^\infty$.
\end{proof}

\noindent
\begin{remark}
The fact that $\Sigma$ is a noncharacteristic surface played an essential role in the proof. If $\Sigma$ were null, the field equations would not uniquely determine derivatives transverse to $\Sigma$ in terms of quantities on $\Sigma$. Consequently, it can be possible to produce nonsmooth solutions by patching smooth solutions along a characteristic surface.
\end{remark}
\medskip

As already mentioned near the beginning of this section, the BSKR wormholes are produced by patching solutions together in the manner described above, so our results apply to the BSKR wormholes. It follows that if the required matching conditions Eqs. (\ref{metmatch})--(\ref{dirmatch}) hold, all of the fields must be smooth when written using our gauge choices. However, the BSKR fields are not smooth. This is most readily seen for the metric, which must be smooth when expressed in Gaussian normal coordinates. The Gaussian normal coordinate $s$ is related to the $r$ coordinate of the metric Eq. (\ref{whmetric}) by
\begin{equation}
    s(r) = \int_0^r F_1 (r) dr
\end{equation}
However, in all of their solutions, $dF_1/dr$ is discontinuous at $r=0$ \cite{PC}. Consequently, $d^2 s/dr^2$ is discontinuous at $r=0$. On the other hand, from the explicit expression (\ref{F2}), it can be seen that $F_2(r)$ is a smooth function of $r$ with $dF_2/dr = 0$ at $r=0$ but $d^2F_2/dr^2 \neq 0$ at $r=0$. Using the chain rule, we obtain
\begin{equation}
    \frac{d^3 F_2}{d s^3} = \frac{d^3 F_2}{d r^3} \left(\frac{d r}{d s} \right)^3 + 3 \frac{d^2 F_2}{d r^2} \frac{d r}{d s} \frac{d^2 r}{d s^2} + \frac{d F_2}{d r} \frac{d^3 r}{d s^3} 
\end{equation}
The first and last terms on the right side are continuous at $r=0$ but the middle term is discontinuous. Thus, we see that the BSKR wormhole metric fails to be $C^3$ at $r=0$. This fact also can be deduced from the plots given in \cite{Blazquez-Salcedo:2020czn} for the scalar curvature and the Kretschmann scalar, which can be seen to have a discontinuous derivative at $r=0$. It follows that the BSKR wormholes cannot satisfy all of the necessary matching conditions Eqs. (\ref{metmatch})--(\ref{dirmatch}). 

The fact that the scalar curvature has a discontinuous derivative at $r=0$ implies that the trace of the stress-energy tensor also has a discontinuous derivative at $r=0$. Since the Maxwell stress-energy tensor has vanishing trace, the Dirac stress-energy tensor must have a discontinuous derivative at $r=0$. This strongly suggests that the Dirac field cannot satisfy the required matching condition (\ref{dirmatch}). In the next section, we will analyze the matching conditions and show that this is the case.

\section{Matching Conditions for BSKR Wormholes}
\label{dirac}

We turn now to an analysis of the matching conditions for BSKR wormholes. This will require some care for the treatment of the Dirac field, so we first review some basic properties of Dirac spinors (see e.g. \cite{Wald:1984rg}).

In terms of Weyl spinors, a Dirac spinor is a pair composed of a Weyl spinor $\phi^A$ and complex conjugate Weyl spinor $\bar{\psi}^{A'}$. We introduce a basis $o^A$ and $\iota^A$ for the Weyl spinor space $W$ satisfying $o_A \iota^A = 1$ and use the complex conjugate basis $\bar{o}^{A'}$ and $\bar{\iota}^{A'}$ for the complex conjugate spinor space $\bar{W}$. From these bases, we can construct the quantities
\begin{align}
  t^{AA^\prime} = \frac{1}{\sqrt{2}}(o^A \bar o^{A^\prime} + \iota^A \bar \iota^{A^\prime}) \\
  x^{AA^\prime} =\frac{1}{\sqrt{2}}( o^A\bar\iota^{A^\prime} + \iota^A\bar o^{A^\prime})\\
  y^{AA^\prime} =\frac{i}{\sqrt{2}} (o^A\bar\iota^{A^\prime} - \iota^A\bar o^{A^\prime}) \\
  z^{AA^\prime} =\frac{1}{\sqrt{2}}( o^A \bar o^{A^\prime} - \iota^A \bar \iota^{A^\prime}).
\end{align}
which can be identified with an orthonormal tetrad $\{ e^a_\mu \}$ in spacetime. Conversely, a choice of orthonormal tetrad $\{e^a_\mu \}$ corresponds to a choice of spin basis $o^A, \iota^A$ up to sign.

In the presence of a vector potential $A_a$, the Dirac equation for the spinors $\phi^A$ and $\bar{\psi}^{A'}$ is given by the pair of equations
\begin{align}
    (i\nabla^{AA^\prime} + q A^{AA'}) \phi_A &= m\bar\psi^{A^\prime} \label{dirac1}\\
    (i\nabla_{AA^\prime} + q A_{AA'}) \bar\psi^{A^\prime}&= m \phi_A \label{dirac2}
\end{align}
We expand $\bar{\psi}^{A'}$ in the basis $\bar{o}^{A'}$ and $\bar{\iota}^{A'}$
\begin{equation}
    \bar{\psi}^{A'} = \alpha \bar o^{A^\prime} + \beta\bar \iota^{A^\prime}
\end{equation}
and we expand $\phi_A$ in the dual spinor basis $o_A^* = -\iota_A$ and $\iota_A^* = o_A$
\begin{equation}
    \phi_A = \gamma o_A^* + \delta \iota_A^*
\end{equation}
The Dirac spinor can then be represented by the components
\begin{equation}
    \Psi_\alpha =  \begin{bmatrix}
\alpha \\
\beta \\
\gamma\\
\delta \end{bmatrix} \label{dircomponents}
\end{equation}
In flat spacetime, we can choose the spinor basis $o^A, \iota^A$ and the corresponding orthonormal tetrad $\{e^a_\mu \}$ to be constant (i.e., have vanishing derivative) over spacetime. The Dirac equations (\ref{dirac1}) and (\ref{dirac2}) then take the form
\begin{equation}
 \gamma^\mu (i \partial_\mu + q A_\mu) \Psi = m \Psi
 \label{dirac3}
\end{equation}
where we have omitted the Dirac spinor indices and where
\begin{align}
    & \gamma^0 = \begin{bmatrix}
    0 & 0 & 1 & 0\\
    0 & 0 & 0 & 1\\
    1 & 0 & 0 & 0\\
    0 & 1 & 0 & 0\end{bmatrix} &&  \gamma^1 = \begin{bmatrix}
    0 & 0 & 0 & 1 \\
    0 & 0 & 1 & 0\\
    0 & -1 & 0 & 0\\
    -1 & 0 & 0 & 0\end{bmatrix} \nonumber \\
& \gamma^2 =\begin{bmatrix}
 0 & 0 & 0 & -i \\
 0 & 0 & i & 0 \\
 0 & i & 0 & 0 \\
 -i & 0 & 0 & 0 \\
\end{bmatrix} &&  \gamma^3 = \begin{bmatrix}
    0 & 0 & 1 & 0\\
    0 & 0 & 0 & -1\\
    -1 & 0 & 0 & 0\\
    0 & 1 & 0 & 0\end{bmatrix}.
\end{align}
This corresponds to the standard form of the Dirac equation in the chiral representation. In curved spacetime (or in a nonconstant basis in flat spacetime), additional terms will arise in (\ref{dirac3}) from the derivatives of the spinor basis, which may be computed in terms of the Ricci rotation coefficients of the corresponding tetrad $\{e^a_\mu \}$.

For the BSKR wormholes, the solution $(g^+_{ab}, A^+_a, \Psi^+_\alpha)$ is taken to be of the following form. The metric is assumed to be given by Eq. ~(\ref{whmetric}), with $M^+$ taken to be the region $r > 0$. The vector potential on $M^+$ is taken to be of the form
\begin{equation}
    A^+_a = V(r) dt
    \label{vpan}
\end{equation}
The Dirac field is taken to be an incoherent superposition of two solutions of the form
\begin{align}
    &&\Psi_{[1]}^+ =\begin{pmatrix}
    \cos \left(\frac{\theta }{2}\right) z(r) e^{i
   \left(\frac{\phi }{2}-t w\right)}\\
   i \kappa  \sin
   \left(\frac{\theta }{2}\right) \bar{z}(r) e^{i
   \left(\frac{\phi }{2}-t w\right)}\\
   -i \cos
   \left(\frac{\theta }{2}\right) \bar{z}(r) e^{i
   \left(\frac{\phi }{2}-t w\right)}\\
   -\kappa  \sin
   \left(\frac{\theta }{2}\right) z(r) e^{i
   \left(\frac{\phi }{2}-t w\right)}
    \end{pmatrix}
   , &&& \Psi_{[2]}^+ = \begin{pmatrix}i \sin \left(\frac{\theta }{2}\right) z(r) e^{i
   \left(-t w-\frac{\phi }{2}\right)}\\
   \kappa  \cos
   \left(\frac{\theta }{2}\right) \bar{z}(r) e^{i \left(-t
   w-\frac{\phi }{2}\right)}\\
   \sin \left(\frac{\theta
   }{2}\right) \bar{z}(r) e^{i \left(-t w-\frac{\phi
   }{2}\right)}\\
   i \kappa  \cos \left(\frac{\theta
   }{2}\right) z(r) e^{i \left(-t w-\frac{\phi
   }{2}\right)}\end{pmatrix} \label{diransatz}
\end{align}
where $\kappa = \pm 1$. An incoherent superposition of this sort is necessary in order to get a total current and stress-energy that is spherically symmetric. BSKR obtain numerical solutions to the EDM equations for $(g^+_{ab}, A^+_a, \Psi^+_\alpha)$ on $M^+$.

An obvious choice for $(g^-_{ab}, A^-_a, \Psi^-_\alpha)$ would be to take it to be an identical copy of $(g^+_{ab}, A^+_a, \Psi^+_\alpha)$. In that case, the matching conditions (\ref{metmatch}) and (\ref{vpmatch}) hold automatically. In order for the matching condition (\ref{ecmatch}) to hold, it is necessary and sufficient for the extrinsic curvature of $\Sigma^+$ to vanish. This holds if and only if
\begin{equation}
     \frac{\partial F_0}{\partial r} \bigg\vert_{r=0} = 0.
     \label{F0match}
\end{equation}
This condition is imposed by BSKR. Similarly, in order for (\ref{vpmatch2}) to hold, it is necessary and sufficient that
\begin{equation}
     \frac{\partial V}{\partial r} \bigg\vert_{r=0} = 0.
\end{equation}
This condition was {\em not} imposed by BSKR \cite{PC}. The failure of this condition to hold implies the presence of an additional charged shell of matter at $r=0$. 

We now consider the matching condition for the Dirac spinor. Equation (\ref{dirmatch}) applies for a continuous choice of tetrad. However, in the matching of $(g^+_{ab}, A^+_a, \Psi^+_\alpha)$ with its identical copy $(g^-_{ab}, A^-_a, \Psi^-_\alpha)$, the tetrad vector $e^{-a}_3$ points in the wrong direction as compared with the continuous tetrad choice made in the previous section. Thus, we must take the discontinuity of the tetrad vector $e^a_3$ at $r=0$ into account when formulating the matching conditions for the Dirac field. To do so, we note that the reversal of the tetrad vector $e^a_3$ at a point $x \in \Sigma$ corresponds to a parity transformation on the tetrad followed by a $180^\circ$ rotation about $e^a_3$. The application of these transformations to the tetrad while keeping the Dirac spinor unchanged is equivalent to applying the inverse of these transformations to the Dirac spinor while keeping the tetrad fixed. The parity operator on Dirac spinors is given by
\begin{equation}
   {\mathcal P} = \eta \gamma^0
\end{equation}
where one can make any of the choices $\eta = \{ 1, -1, i, -i \}$. The inverse transformation is of the same form. A $180^\circ$ rotation about $e^a_3$ is given by\footnote{We assume here that the $r$ direction in the BSKR Dirac spinor ansatz corresponds to what we are calling the $3$ direction. (BSKR do not explicitly say this, but any other choice would give rise to inconsistencies in the angular dependence of the quantities being matched.)}
\begin{equation}
    \mathcal{R} =\pm \left(
\begin{array}{cccc}
 i & 0 & 0 & 0 \\
 0 & -i & 0 & 0 \\
 0 & 0 & i & 0 \\
 0 & 0 & 0 & -i \\
\end{array}
\right).
\end{equation}
and the inverse also is of the same form. Thus, for the case where $(g^-_{ab}, A^-_a, \Psi^-_\alpha)$ is an identical copy of $(g^+_{ab}, A^+_a, \Psi^+_\alpha)$, the matching condition (\ref{dirmatch}) for a Dirac field becomes
\begin{equation}
    \Psi^+ \big\vert_{r=0} = {\mathcal P}^{-1} {\mathcal R}^{-1} \Psi^- \big\vert_{r=0} = {\mathcal P}^{-1} {\mathcal R}^{-1} \Psi^+ \big\vert_{r=0}
\end{equation}
where for the last equality we have used the equality of $\Psi^-$ and $\Psi^+$ in the original tetrads at $r=0$. For the ansatz (\ref{diransatz}), this matching condition for $\Psi^+_{[1]}$ yields
\begin{equation}
 \eta \begin{bmatrix}
-\cos \left( \frac{\theta }{2}\right) e^{\frac{i \phi }{2}-i t w}
   \bar{z}(0) \\
   -i \kappa  \sin \left(\frac{\theta }{2}\right)
   e^{\frac{i \phi }{2}-i t w}   z(0) \\
   -i \cos \left(\frac{\theta
   }{2}\right) e^{\frac{i \phi }{2}-i t w} z(0)\\
   -\kappa  \sin
   \left(\frac{\theta }{2}\right) e^{\frac{i \phi }{2}-i t w} \bar{z}(0) 
\end{bmatrix}=\begin{bmatrix}
-\cos \left(\frac{\theta }{2}\right) e^{\frac{i \phi }{2}-i t w}
   z(0)\\
   -i \kappa  \sin \left(\frac{\theta }{2}\right)
   e^{\frac{i \phi }{2}-i t w} \bar{z}(0)\\
   i \cos \left(\frac{\theta
   }{2}\right) e^{\frac{i \phi }{2}-i t w} \bar{z}(0)\\
   \kappa  \sin
   \left(\frac{\theta }{2}\right) e^{\frac{i \phi }{2}-i t w} z(0)
\end{bmatrix} \\
\label{psi1match}
\end{equation}
for some choice of $\eta = \{ 1, -1, i, -i \}$. It is easily seen that this condition cannot be satisfied for any choice of $\eta$ unless $z(0) = 0$. However, the Dirac fields obtained by BSKR have $z(0) \neq 0$. Thus, the wormhole spacetimes obtained by taking $(g^-_{ab}, A^-_a, \Psi^-_\alpha)$ to be an identical copy of $(g^+_{ab}, A^+_a, \Psi^+_\alpha)$ have a spurious delta-function source term for the Dirac field at $r=0$. As noted above, they also have a spurious charged shell of matter at $r=0$.

One could also consider other possible choices of $(g^-_{ab}, A^-_a, \Psi^-_\alpha)$ that are related to $(g^+_{ab}, A^+_a, \Psi^+_\alpha)$ by symmetry operations that map $\Sigma$ to itself. It would appear that the only potentially viable option would be a time reflection operation. For the metric (\ref{whmetric}), time reflection takes $g^+_{ab}$ to $g^+_{ab}$, but for the vector potential (\ref{vpan}), it takes $A^+_a$ to $-A^+_a$. Thus the spacetime obtained joining $(g^+_{ab}, A^+_a, \Psi^+_\alpha)$ to its time reflection would have $V(-r) = - V(r)$. Since the vector potential, $A_a$, flips sign between $r > 0$ and $r < 0$, it follows from Maxwell's equations that the charge-current vector $J^a$ of the Dirac field must correspondingly flip sign. However, the probability $4$-current, $P^a$, of any Dirac spinor is always a future-directed timelike vector, and the charge-current vector is given by $J^a = q P^a$. Therefore, in order to construct a wormhole spacetime in this manner, the Dirac field for $r < 0$ must have charge that is of opposite sign\footnote{This fact is undoubtedly closely related to points raised in \cite{Bronnikov:2021jlz} and \cite{Konoplya:2021hsm}.} to that of the Dirac field for $r > 0$. It is unclear to us what interpretation could be given to a quantity obtained by combining a Dirac field of charge $q$ for $r > 0$ with a Dirac field of charge $-q$ for $r < 0$, but it is clear that such a quantity cannot in any sense be considered to be a classical solution to the Dirac equation on the wormhole spacetime.

A final possibility along these lines would be to take $g^-_{ab} = g^+_{ab}$ and $A^-_a = A^+_a$, but take $\Psi^-_\alpha = {\mathcal T} \Psi^+_\alpha$, where $\mathcal T$ is the Dirac time reversal operator. For a vector potential of the form (\ref{vpan}), ${\mathcal T} \Psi^+_\alpha$ will satisfy the Dirac equation with the original $q$. The time reversal operator reverses the spatial components of the Dirac charge-current $J^a$ and the time-space components of the Dirac stress-energy tensor $T_{ab}$, but since these components vanish in the original solution $(g^+_{ab}, A^+_a, \Psi^+_\alpha)$, it follows that $(g^+_{ab}, A^+_a, {\mathcal T} \Psi^+_\alpha)$ will satisfy the EDM equations. For this choice of $(g^-_{ab}, A^-_a, \Psi^-_\alpha)$, the metric matching conditions will again be satisfied, provided that (\ref{F0match}) has been imposed. Since $A_a(-r)=A_a(r)$, the Maxwell matching condition (\ref{vpmatch2}) again does {\em not} hold, thereby requiring an additional shell of charged matter at $r=0$. We now consider the matching conditions for the Dirac field.

The action of the time reversal map $\mathcal T$ on a Dirac spinor (\ref{dircomponents}) is 
\begin{equation}
  {\mathcal T}  \begin{pmatrix}
\alpha(t) \\
\beta(t) \\
\gamma(t)\\
\delta(t)
\end{pmatrix} = e^{i \rho }\begin{pmatrix}
\bar\beta(-t) \\
- \bar\alpha(-t) \\
 \bar\delta(-t)\\
-\bar\gamma (-t)
\end{pmatrix}
\end{equation}
where $e^{i \rho}$ is an arbitrary phase. The required matching condition is now
\begin{equation}
    \Psi^+ \big\vert_{r=0} = {\mathcal P}^{-1} {\mathcal R}^{-1} \Psi^- \big\vert_{r=0} = {\mathcal P}^{-1} {\mathcal R}^{-1} {\mathcal T} \Psi^+ \big\vert_{r=0}.
    \label{dirmatchT}
\end{equation}
If we apply this condition to $\Psi^+_{[1]}$ using the ansatz (\ref{diransatz}), we find that the angular factors do not match, so (\ref{dirmatchT}) cannot be satisfied. Nevertheless, we can instead try to match $\Psi^+_{[1]}$ to ${\mathcal P}^{-1} {\mathcal R}^{-1} \Psi^-_{[2]} = {\mathcal P}^{-1} {\mathcal R}^{-1} {\mathcal T} \Psi^+_{[2]}$ at $r=0$. In this case, the angular factors do match. However, a calculation similar to that of Eq. (\ref{psi1match}) shows that the required matching condition holds only when $z(0)= 0$, which is not satisfied by any of the BSKR wormholes.

In summary, our analysis of the matching conditions shows that the BSKR wormholes require the presence of an additional shell of charged matter and, more seriously, contain a spurious distributional source for the Dirac field at $r=0$. This confirms the conclusion of Sec. \ref{smoothness} that the BSKR wormholes are not solutions to the EDM equations.

\medskip 

\noindent
{\bf Acknowledgements} We thank Jose Blázquez-Salcedo, Christian Knoll, and Eugen Radu for providing us with \cite{PC}---as well as an earlier draft of \cite{PC}---prior to posting. D.L.D. acknowledges his support as a Fannie and John Hertz Foundation Fellow, holding the Barbara Ann Canavan Fellowship. The research of R.J.W. was supported by the National Science Foundation Graduate Research Fellowship under Grant No. DGE 1706045. This research was supported in part by NSF Grant NJo. 21-05878 to the University of Chicago.

\bibliographystyle{JHEP.bst}
\bibliography{main.bib}

\end{document}